\documentclass[11pt]{article}
\pdfoutput=1

\usepackage{times}
\usepackage{fullpage}

\usepackage[noend]{algorithmic}
\usepackage{graphicx}
\usepackage{wrapfig}
\usepackage{amsfonts} 

\newcommand{\turkoglu}{T\"{u}rko\u{g}lu}

\newcommand{\call}[1]{\textsc{#1}}
\newcommand{\Real}{\mathbb{R}}

\newcommand{\xor}{\mathop{\textbf{xor}}}
\newcommand{\NN}{\mathrm{NN}}
\newcommand{\polylog}{\mathrm{polylog}}

\newtheorem{definition}{Definition}[section]
\newtheorem{lemma}[definition]{Lemma}

\newtheorem{theorem}[definition]{Theorem}

\def\EndProof{ \quad \vrule width 1ex height 1ex depth 0pt }
\newenvironment{proof}{\textbf{Proof:}\hspace{4pt}}{\EndProof}

\newenvironment{icompact}{
  \begin{list}{$\bullet$ }{
    \leftmargin .25in 
    \setlength{\itemindent}{-\labelwidth-\labelsep-2pt}
    \parsep 0pt
    \partopsep 0pt
    \topsep 2pt plus 2pt minus 2pt
    \itemsep 0pt}}%
  {\end{list}}

\title{Succinct Representation of Well-Spaced Point Clouds}
\author{
Beno\^{\i}t Hudson \\
Toyota Technological Institute at Chicago
}
\date{}

\begin{document}

\maketitle

\begin{abstract}
A set of $n$ points in low dimensions takes $\Theta(nw)$ bits to store
on a $w$-bit machine.  Surface reconstruction and mesh refinement
impose a requirement on the distribution of the points they process.
I show how to use this assumption to lossily compress a set of $n$
input points into a representation that takes only $O(n)$ bits,
independent of the word size.  The loss can keep inter-point distances
to within 10\% relative error while still achieving a factor of three
space savings.  The representation allows standard quadtree
operations, along with computing the restricted Voronoi cell of a
point, in time $O(w^2 + \log n)$, which can be improved to time
$O(\log n)$ if $w \in \Theta(\log n)$.  Thus one can use this
compressed representation to perform mesh refinement or surface
reconstruction in $O(n)$ bits with only a logarithmic slowdown.

\end{abstract}

\section{Introduction}

My goal in this paper is to produce a succinct point location structure to
support surface reconstruction and mesh refinement tasks while using
asymptotically less space than the current state of the art.  To this end, I
present a compressed data structure that provides an interface much like that
of a $2^d$-tree (for readability I call this a quadtree).  Funke and
Milosavljevic~\cite{funke07network} show
how to use an quadtree so that given a sufficiently dense set of points in
3-space that lie on an unknown 2-manifold, we can reconstruct a triangulation
that approximates the 2-manifold, in time $O(n \log n)$.  Hudson and
\turkoglu{}~\cite{hudson08efficient} show how to use a quadtree
to augment an arbitrary set of points in $d$ dimensions and generate
a well-spaced point set, that is, one whose Delaunay triangulation will produce
a good mesh for scientific applications, usually in $O(n \log n)$ time (more
precisely, in $O(n \log \Delta)$ time where $\Delta$ is the geometric spread).
As a special case, if the point set is well-spaced, they compute the Voronoi
cell of an arbitrary point in constant time.

The storage used by the two algorithms is in two parts: the quadtree, and the
point coordinates.  A balanced quadtree~\cite{bern94provably} over $n$ arbitrary
points can have a superlinear number of quadtree cells.  However, the
requirements of the tasks at
hands impose limits on the size of the quadtree.  In the surface reconstruction
case, the $n$ points of the input must be \emph{locally regular} (which has
various definitions,
see Section~\ref{sec:prelims}).  In the mesh refinement case, the output
must be well-spaced, and I use $n$ to refer to the number of points in the
output.  These requirements imply that the quadtree only has
$O(n)$ leaves and internal nodes~\cite{hudson09size}.  Thus, in a cell probe
model of computation with real arithmetic, the storage is linear.  Sadly, this
is a completely fictional machine.  In the more realistic word RAM model,
each pointer must be at least $\log(n)$ bits long for each point to have a
distinct memory address, whereas each coordinate must be
at least $\log(n^{1/d})$ bits for each point to have a distinct location in
space.  For simplicity I assume both coordinates and pointers have the same
length, namely $w$ bits.

The first part of this paper
(Sections~\ref{sec:prelims}--\ref{sec:morton}) lays the groundwork,
showing that we can use a small interface on top of a quadtree in order to
support the richer operations that Funke \emph{et al} or Hudson \emph{et al}
require.  I show how to implicitly represent the balanced quadtree by sorting
the coordinates in Morton order~\cite{morton66computer} (also known as the
Z-order), a trick that extends the so-called \emph{framework} of Bern,
Eppstein, and Teng~\cite{bern99parallel}.  This all but eliminates the
storage cost of the quadtree, while only requiring a logarithmic slowdown.

The second and more interesting part of the paper
(Sections~\ref{sec:diffcode}--\ref{sec:fast}) deals with compressing the
\emph{geometric}
storage requirements at asymptotically no greater runtime cost.  The intuition
is that, given that we have the points in sorted order, the distance between
one point and the next is generally small, therefore the difference in
coordinates is small.  What my structure stores for each vertex is the
coordinates of each point, rounded to be at a fixed position in its quadtree
leaf square: the representation is lossy, as it must be
(Lemma~\ref{lem:rounding-needed}).
It keeps track of the
number of bits that were rounded off by storing the quadtree leaf size.
Intuitively, the leaf sizes generally do not change much from one point to the
next, so it only stores the difference in the leaf sizes.  Intuitively, only the
low bits of the coordinates of each point change from one point to the next in
the sorted order, so it only stores the XOR of the coordinates.  To prove these
intuitions I show a correspondence between the storage requirements and a
depth-first traversal of the balanced quadtree: since we can traverse an
$n$-node tree by visiting each node only once, we can store the $n$ geometric
points in $O(n)$ bits.  I discuss very preliminary experimental results
in the conclusions.

Traditional compression techniques read in the entire input into memory and
compress it, the intent being to then transmit the file over the network, or to
save on disk storage.  But the entire point of my work here is to handle inputs
that are too big for main memory, without switching to an out-of-core or
streaming framework.  I show instead how to read points in from a file, without
exceeding the space bounds, using $w$ scans through the file for a total of
about $O(n \log^2 n)$ time to initialize the structure assuming words of length
about $\log n$.

\paragraph{Prior work:}
Traditional work on mesh compression store both the mesh topology and the
mesh geometry using a 10--20 bits per vertex, most of it geometric
information~\cite[survey]{peng05technologies}.  Unfortunately, this is
inapplicable since in the setting of this paper, we do not have
connectivity information: in the reconstruction case, connectivity is the
final goal, and in the refinement case, connectivity is too expensive to
maintain.

Another related approach is a succinct mesh structure by Blandford \emph{et
al}~\cite{blandford05compact, blandford05thesis}, which can be used to
compute the Delaunay triangulation of a set of points.  They compress the
connectivity as they discover it, but do not compress the geometry.  Thus
our approaches are complementary: I show how to compress the geometry, but
not the connectivity.  It remains future work to see how precisely to meld
the two ideas.

\section{Preliminaries}
\label{sec:prelims}

\paragraph{Machine Model:}
I assume we operate on the $w$-bit word RAM.  That is, pointers are $w$-bit
memory addresses, and coordinates consume $w$ bits also.  For simplicity, I
assume the points are at positive fixed-point (\emph{i.e.}\ integer)
coordinates; it will be clear that the sign bit and exponent of floating point
notation can easily be handled.  Coordinates thus go from 0 to $W \equiv 2^w -
1$.

\paragraph{Input Assumption:}
Given a point set $P$ and subspace $S$, the \emph{restricted Voronoi
diagram} assigns to each $p \in P$ a cell $V_p$ consisting of the set of points
$x \in S$ that are closer to $p$ than to any other point $q \in P$.  In other
words, it is the intersection of the Voronoi diagram of $P$ and the subspace
$S$.  The \emph{restricted Delaunay neighbours} of $p \in P$ are the points
of $Q$ whose restricted Voronoi cells intersect $V_p$.
The results herein assume the restricted Voronoi cells are well-spaced in the
following sense:

\begin{definition}
The \textbf{aspect ratio} of $V_p$ (and, by extension, of $p$) is the ratio
$\frac{\max_{x\in V_p} |px|}{\min_{q \in P} |pq|}$.  We say that $P$
is \textbf{$\rho$-well-spaced} if for all $p \in P$, the aspect ratio of $V_p$
is at most $\rho$.
\end{definition}

In the mesh refinement problem, one of the goals is precisely to generate a
well-spaced set of points.  In surface reconstruction, there are many
assumptions that can be required of the input.  In particular, many algorithms
are only guaranteed to work if the input is an $\epsilon$-net in the following
sense:

\begin{definition}
Let $f : S \to
\Real$ be a Lipschitz function; that is, $f(x) \leq f(y) + |xy|$ for any $x$
and $y$ both in $S$.  Then the $n$-point set $P \subset S$ is an
\textbf{$\epsilon$-net} over $S$ if:
\begin{icompact}
\item for all $x \in S$, there is a point $p \in P$ such that $|px| \leq
	\epsilon f(x)$
\item for all $p \in P$, for all $q \in P \backslash \{p\}$, $|pq| \geq f(p)$.
\end{icompact}
\end{definition}

An $\epsilon$-net is well-spaced.  Consider a restricted Voronoi cell $V_p$.
The denominator in the aspect ratio is at least $f(p)$ by definition.  For the
numerator, consider a point $x \in V_p$.  We know that $|px| \leq \epsilon
f(x)$.  By the Lipschitz assumption on $f$, we know $f(x) \leq f(p) + |px|$.
Thus $|px| \leq f(p) / (1-\epsilon)$.  In other words, the aspect ratio is at
most $\epsilon / (1 - \epsilon)$.

\section{Quadtree Query Structure}
\label{sec:quadtree}

The query structure provides an interface that the balanced
quadtree~\cite{bern94provably} can support.  See Figure~\ref{fig:api}.

\begin{figure}
\centering
\begin{tabular}{rl}
$\call{SquareOf}(p)$	& Return the leaf that contains the point $p$ \\
$\call{Vertices}(s)$	& Return a pointer to the set of vertices in the square $s$ \\
$\call{Neighbour}(s, i)$& Return the $i$th equal-sized neighbour of the square $s$ \\
$\call{Child}(s, i)$	& Return the $i$th child of the square $s$ \\
$\call{Voronoi}(p)$	& Compute the restricted Voronoi cell of $p$ \\
\end{tabular}
\caption{
\label{fig:api}
Interface I support for the quadtree.  
\call{Voronoi} is implemented in terms of the first four queries.  All run in constant time for a pointer-based quadtree.}
\end{figure}

In my parlance, a quadtree is a tree structure defined recursively: each node,
which I call a \emph{square}, has a size and a defining point, namely its
minimum corner.  Internal squares have exactly $2^d$ children.  Leaf squares
store
a list of the points they contain.  We can recursively define the
\emph{balanced quadtree} of Bern, Eppstein, and Gilbert~\cite{bern94provably}:
we start with the root square, with minimum corner at the origin and size
$W$.  We split the largest quadtree leaf square that is either crowded or
unbalanced.  A square is crowded if either it contains two or more
vertices, or it contains exactly one vertex, and a neighbouring square contains
one or more vertices.  A square is unbalanced if there is a neighbouring
square of one quarter the size.  Splits are in the middle of the node: the
children are non-overlapping and have half the size of their parent.
The main obvious property we need here is that, in the balanced quadtree, each
vertex is assigned to exactly one leaf square, and all the equally-sized
neighbouring squares are empty of any vertices.

The traditional representation has each square storing $2^d$
pointers to children, $3^d-1$ pointers to neighbouring equal-sized squares, 
the coordinates, and the height (the logarithm of the size) of the square.
Leaf squares store a pointer to the
set of vertices they contain.  Each
point stores a pointer to the square in which it lies.  The implementations of
most of the functions are fairly obvious given this representation.  
Insertion and deletion are also relatively straightforward.
They take time linear in the depth of the tree; given that the points are at
integer coordinates, the depth of the tree is at most $w$.
I should draw attention to the fact that $\call{Vertices}$ is not required to
report a copy of the set of vertices in the square, but merely pointers that
can be used to iterate over the set.  

The restricted Voronoi query requires a bit more discussion.  
If the subspace to which we are restricting is the entire volume $[0,W]^d$,
Hudson and \turkoglu{}~\cite[Theorem 7]{hudson08efficient} showed how to
compute the restricted Voronoi cell of a point assuming the point set is
well-spaced:
First, look up the square that contains the query point $q$.  Then, scan the
squares, in order of distance from $q$.  Vertices encountered during the scan
are stored as the putative Voronoi neighbours.  Using the appropriate notion of
distance, this procedure indeed computes the exact Voronoi cell, and only
accesses a bounded number of squares.

For an unknown 2-manifold $S$ lying in arbitrary dimension, assuming a
\emph{locally regular sample} (which is a more stringent requirement than the
$\epsilon$-net condition) we can use the procedure defined by Funke and
Milosavljevic~\cite{funke07network} to compute in constant time (looking up a
bounded number
of squares) the Voronoi cell restricted to a 2-manifold $S'$ that approximates
$S$.  For higher-dimensional manifolds, the situation is still open, though we
know how to compute a bounded-size superset of the restricted Voronoi
cell~\cite{hudson08towards} that still needs to be
pruned~\cite{cheng05manifold}.  Nevertheless, in all cases, the
Voronoi query is in two parts: (a) look up the balanced quadtree square $s$ that
contains the query point.  (b) look up a bounded number of nearby quadtree
squares, each of them no more than a constant factor smaller than $s$.

\section{Implicit Quadtree via Morton Ordering}
\label{sec:topo}
\label{sec:morton}

The first new result in this paper is to show how to nearly eliminate the cost
of storing the quadtree.
The key observation is that the balanced quadtree is formed from the leaves
of a trie with arity $2^d$.  To
recur into a split child is equivalent to reading one more bit of each
coordinate.
A trie square $s$ is represented by its minimum corner $p$ and a height
$h$.  It must be that $p$ is a prefix of length $d (w - \log h)$, and
the lower $\log h$ bits of each coordinate are all zero; otherwise, $s$ is
not a square of the trie.  To split a trie square, simply halve $h$ and use
one additional bit of each coordinate in $p$.  To find a neighbour of a
trie square, add $h$ to the appropriate coordinates of $p$.

This motivates computing the Morton Order~\cite{morton66computer} (also known
as the Z-order).  Define the Morton number of a point as the $(d \log n)$-width
integer formed by interleaving the bits of the coordinates, see
Figure~\ref{fig:morton}.  The Morton order of a set $P$ sorts by the Morton
number.  The data structure representing the quadtree is then simply an array:
we explicitly store the coordinates of the points.  I now show how to
implement the quadtree interface of the prior section:

\begin{figure}
\centering
\begin{tabular}{lccc}
Decimal: & (3, 5)	   &     & (4, 2) \\
Binary:  & (011, 101) 	   &     & (100, 010) \\
Morton:  & 01 10 11	   & $<$ &  10 01 00 
\end{tabular}
\caption{\label{fig:morton} The Morton order of points (3, 5) and (4, 2).}
\end{figure}

\paragraph{$\call{Vertices}(s)$:}
We can look up the contents of a trie square $s = (p, h)$: search for the
successor of the minimum corner of $s$ (namely, $p$) in the Morton order,
and search for the predecessor of the maximum corner of $s$ (namely, $p$
translated by $h$ in each coordinate).  Every point strictly within this
range has a prefix that places it within the square of interest.  
The two binary searches take $O(\log n)$ time each and return the bounds of a
subarray over which we can iterate.

\paragraph{\call{Neighbour} and \call{Child}:}
The representation of a square now is merely a prefix and a size, which allows
computing these functions in constant time using mere arithmetic.  To get an
equal-sized neighbour of the square, merely add or subtract the size from each
coordinate.  To get the child of a square, append $d$ bits to the prefix and
halve the size.

\paragraph{$\call{SquareOf}(p)$:}
The leaf square that contains $p$ and is uncrowded shares a common prefix with
$p$; the only question is how much of the prefix needs to be kept.  The answer
can be determined by binary search over the $w$ possible answers.
Given a trie square $s$, we test whether the square is crowded.
First, do a range search for $s$.  If it indicates a range of two or more
vertices, then $s$ is crowded.  If it indicates an empty range, $s$ is
uncrowded.  Otherwise, if the range of $s$ contains exactly one vertex,
iterate over each neighbour $s'$ of $s$ and do a range search on $s'$.  If
$s'$ is non-empty, $s$ is crowded, otherwise is it uncrowded.  There are
$3^d - 1$ neighbours of $s$, so testing for crowding takes $O(\log n)$ time.
We do this $O(\log w)$ times, for a total of $O(\log(n) \log(w))$ runtime.

\begin{theorem}
Given a well-spaced point set $P$ in an array sorted by Morton order, we
can compute the restricted Voronoi cell of a given vertex $v$ in $O(\log(n)
\log(w))$ time.
\end{theorem}

\section{Lossless Compression}
\label{sec:diffcode}

Given a list of integers $i_1, \ldots, i_n$, we can often reduce the storage
requirement of the set (normally $wn$ bits) by difference-encoding the
integers, assuming subsequent integers are on average close.  The storage
representation is as follows: let $\delta_j \equiv i_j - i_{j-1}$.  We store
the first element $i_1$ in longhand, using $w$ bits.  For the $j$th element, we
store only $\delta_j$ using a variable bit length encoding of $\delta_j$.
For concreteness, I assume we use a gamma-code: first, a count in unary (using
zeroes) of the number of bits needed to represent $\delta_j$, then $\delta_j$
starting at its leading '1' bit, and finally a sign bit (not necessary if
$\delta_j$ is zero).  The gamma code uses 1 bit if $\delta_j = 0$, or $1 + 2
(\lg |\delta_j| + 1)$ for non-zero values.  Thus, if on average $\delta_j$ is
$O(1)$, the set takes only $w + O(n)$ bits to store.  Note that instead of
using the subtraction operator, we could have used the bitwise operator XOR and
achieved the same asymptotic result.  I refer to this encoding strategy as the
\textbf{difference-code} when using the subtraction operator, and the
\textbf{xor-code} when using bitwise exclusive or.

\begin{figure}
\begin{tabular}{|c|cc|crc|}
\hline
value	 & uncompressed & bits  & xor 		& xor-coded & bits  \\ \hline
$(5, 2)$ & 00101,00010  & 10    & 	       	& 00101,00010 & 10 \\
$(6, 3)$ & 00110,00011	& 10    & (11, 1)      	& 0011,01   & 6 \\
$(8, 4)$ & 01000,00100  & 10    & (1110, 111)  	& 00001110,000111 & 14 \\
$(9, 6)$ & 01001,00110	& 10    & (1, 10)      	& 01,0010   & 6 \\
$(10, 6)$& 01010,00110	& 10    & (11, 0)      	& 0011, 1   & 5   \\ \hline
total	 &		& 50	&	  	&	    & 41 \\ \hline
\end{tabular}
\caption{\label{fig:xor-code} The encoding of a list of 5 points on a 5-bit
machine.  On a 32-bit machine the encodings would take 320 and 95 bits respectively.}
\end{figure}

My encoding strategy applies the xor-code to each dimension in turn.  That is,
first, sort the points by Morton order.  Second, xor-code the array consisting
of all their $x$ coordinates.  Third, xor-code their $y$ coordinates, and so
on.  See Figure~\ref{fig:xor-code}.  I use the xor-code rather than the
difference-code for an easier analysis; it happens to also be very slightly
more space-efficient in practice.

\begin{theorem}
\label{thm:uniform-compression}
Let $f$ be a constant function defined on a compact space $S$; that is, there is
a constant $f_0$ such that $f(x) = f_0$ for all $x \in S$.  Let $P$ be a set of
$n$ points that forms an $\epsilon$-net over $S$ with respect to $f$, for some
constant $\epsilon$.  Then we can store $P$ using $O(w + n)$ bits, where the
constant depends only on $f_0$, $\epsilon$, and the dimension $d$.
\end{theorem}

\begin{proof}
The proof is via an equivalence to the quadtree.  Consider the full trie---all
squares representable with $d$ integers of $w$ bits each.  A point $p$ denotes
a path from root to leaf in this trie.  The representation of the difference between
$p_{i-1}$ and $p_i$ denotes a path from the first point to the second.  For
simplicity, assume all the coordinates have equal-length XORs (if not, we store
only less than is analyzed here).  A 0 bit of the length field in each of the $d$
xor-encoded dimensions corresponds to moving up one level in the quadtree.
A bit in the xor in each of the $d$ dimensions corresponds to choosing a child
into which to travel.  Thus the total number of bits of the encoding of the
difference between $p_{i-1}$ and $p_i$ is precisely $d$ times the length of
the path.

The Morton order is a depth-first ordering of the points (the non-empty
leaves); therefore, the overall path length of traversing the points in Morton
order is linear in the number of nodes of the trie (internal or leaf) that
contain at least one point.  I divide this set of nodes in two: first,
non-empty nodes not in the balanced quadtree over $P$; second, non-empty nodes
that are indeed in the balanced quadtree.  We know by the $\epsilon$-net
condition that there is a point within distance $2\epsilon f_0$ of $p$, so
$s_p$ has size at most $\epsilon f_0$ or else the cell is crowded (which would
contradict that it is a leaf).  Thus the first category contains at most $\lg
(2 \epsilon f_0) n = O(n)$ nodes.  To count the second category, consider the
smallest quadtree cell $s_S$ that contains all of $S$.  A theorem of Hudson,
Miller, Phillips, and Sheehy shows that within $s_S$, a volume mesh such as the
balanced quadtree over the $(1-\epsilon)^{-1}$-well-spaced set $P$ has $O(n)$
leaves~\cite{hudson09size}.  The number of nodes in a $2^d$-tree is only a
constant fraction larger than the number of leaves.  We are still left with
counting the ancestors of $s_S$; there are at most $w$ of them.

Putting the arguments together, we see that the compressed representation takes
precisely $dw$ bits to store the first point, plus $O(n + w)$ bits to store the
subsequent points, assuming $f_0$, $\epsilon$, and the dimension $d$ are constants.
\end{proof}

\section{Lossy Compression}

The lossless compression of the previous section only provably works if the
points are on average at a distance near unity from their predecessor.  While
many data sets have this behaviour (for example, the Stanford bunny data set,
scaled by $10^7$), so far our technique does not handle well-spaced point sets
as advertised.  For example, in one dimension, the points $1, 2, \ldots, 2^n$
are well-spaced but will require $n \log n$ bits in the lossless compression
format.  This motivates developing what amounts to a compressed floating-point
format.

More devastatingly, we can develop an information-theoretic lower bound on our
ability to compress data.  The adversary can take an arbitrary well-spaced
subset in $w/2$ bits, then append a '0' to each integer followed by $w/2-1$
bits of random noise.  This noised set is only slightly less well-spaced than
the original: the '0' keeps vertices from becoming arbitrarily close to each
other.  Information theoretic arguments show that with high probability, the
adversary has won: no code can store $o(w)$ bits and represent the random noise
exactly.  The lower bound here does not disprove the claim of
Theorem~\ref{thm:uniform-compression}: $f_0$ is $2^{w/2}$, which is not a
constant factor.  This means we are forced to round the input.

\begin{lemma}
\label{lem:rounding-needed}
No data structure can exactly store an arbitrary well-spaced set of $n$ points
in $[0,W]^d$ using $o(nw)$ bits.
\end{lemma}

\begin{wrapfigure}{r}{8.5cm}
\centering
\includegraphics[width=8cm]{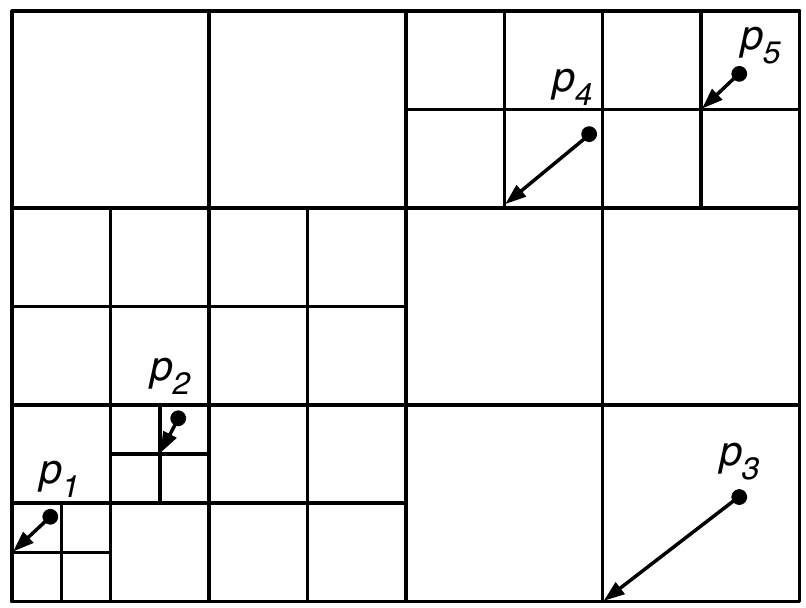}
\caption{\label{fig:rounding} Balanced quadtree of five points in Morton order.  Each point has
an arrow showing its rounded form when $\gamma=0$.}
\end{wrapfigure}
\paragraph{Rounding:}
The rounding routine takes a parameter $\gamma$, which corresponds to the
desired precision: greater $\gamma$ is less lossy, but requires storing more
information.  Consider a point $p$ whose balanced quadtree leaf is at height
$h$.  To round, zero out the lower $h - \gamma$ bits of each coordinate of $p$,
producing a point $p'$.  Intuitively it should be clear that with sufficiently
large $\gamma$, we can maintain any desired well-spacing or $\epsilon$-net
property, so that rounding has little effect.  The key notion is that the
balanced quadtree over the rounded set $P'$ is identical to the balanced
quadtree over the original set $P$.  See Figure~\ref{fig:rounding}.

\begin{lemma}
\label{lem:rounding-error}
Given a pair of input points $p$ and $q$, after rounding them to $p'$ and $q'$,
both the ratio $|p'q'|/|pq|$ and its reciprocal are bounded by $1 + 2^{1-\gamma}\sqrt{d}$.
\end{lemma}

\begin{proof}
By the triangle inequality, $|p'q'| \leq |pq| + |pp'| + |qq'|$.
Without loss of generality, assume $p$ lies in the larger quadtree leaf, and
that this leaf has size $s$.  Rounding moves the points by at most
$2^{-\gamma}\sqrt{d}s$.  Thus $|p'q'| \leq |pq| + 2^{1 - \gamma}\sqrt{d} s$.
Finally, the definition of crowding implies that $s \leq |pq|$, which proves
$|p'q'| \leq (1 + 2^{1 - \gamma}\sqrt{d}) |pq|$.  The proof for the reciprocal
is symmetric.
\end{proof}

\paragraph{Encoding:}
Since we have rounded off some bits, it would be wasteful now to represent all
the zeroes.  Instead, we will store only the number of bits we rounded off ---
that is, the height $h$.  Note that this means we are essentially storing a
floating-point number.  Furthermore, from one point to the next, the height
differs only slightly, so I recommend to difference-code the heights.  Thus,
assuming $\gamma = 0$, the two-dimensional point $(x_i, y_i)$ whose balanced
quadtree leaf is at height $h_i$ is represented as the triple of integers $(h_i
- h_{i-1}, (x_i \xor x_{i-1}) >> h_i, (y_i \xor y_{i-1}) >> h_i)$ where $>>$
denotes a logical right shift.

To decode a point $\mathbf{x}_i$ knowing $\mathbf{x}_{i-1}$ and $h_{i-1}$,
we decode $\partial h_i$ and add $h_{i-1}$ to obtain $h_i$.  Then we decode
each $\delta \mathbf{x}_{i,j}$ in turn, shift $\mathbf{x}_{i-1, j}$ right by
$h_i - \gamma$ bits, xor with $\delta \mathbf{x}_{i,j}$, then shift this
quantity left by $h_i - \gamma$.

\begin{theorem}
\label{thm:lossy-compression}
Let $P$ be a set of $n$ $\rho$-well-spaced points.  Round $P$ to $P'$ with the
rounding coefficient $\gamma$.  We can store $P$ using $O(w + n\gamma)$ bits,
where the constant depends only on $\rho$ and the dimension $d$.
\end{theorem}

\begin{proof}
As in the prior theorem, we use the correspondence with the quadtree.  The bits
we store for the coordinates are paid for by the depth first search argument as
before, except for the $\gamma$ bits we store for additional precision.  It
remains to count the difference-encoded bits for the height.

Consider the following string we could compute during a depth first traversal
through the balanced quadtree: store a $+$ when the search heads to a parent,
and a $-$ when it moves to a child.  The total number of symbols thus stored is
bounded by the number of nodes of the quadtree, which is $O(n)$ for a
well-spaced set of points.  I claim the height differences are essentially a
summary of this same information: indeed, the height difference between $p_i$
and $p_{i+1}$ is the number of $+$ symbols seen when moving from the quadtree
leaf that contains $p_i$ to the least common ancestor with $p_{i+1}$, minus the
number of $-$ symbols from the least common ancestor to the leaf that contains
$p_{i+1}$.  This is clearly more parsimonious (or at worst no worse) than
run-length-encoding the string of $+$ and $-$ symbols.  In the worst case,
run-length-encoding takes $O(n)$ bits when the lengths are stored using the
gamma code.
\end{proof}

\section{Fast Queries}
\label{sec:fast}

The structure we have discussed so far takes only $O(n + w)$ bits to store $n$
well-spaced points, as promised, but it does not fulfill the promise of on-line
queries in sublinear time.  In this section I show how to add a standard
(that is, uncompressed) balanced binary tree over the structure to achieve the
faster runtime.

Recall that a binary search tree over $N$ elements requires $O(N)$ pointers,
which is $O(Nw)$ bits.  This means we can store as many as $n/w$ elements in
the tree and still only use $O(n)$ bits.  The solution, then, is to break the
encoded array of well-spaced points into subarrays of size $n/w$.  In sum,
encoded, the array took $O(n + w)$ bits.  Broken up, we need to pay $O(w)$ bits
to encode the first point in each subarray in longhand (which forms the key in
the search tree).  There are $O(n/w)$ subarrays, so this overhead amounts to
$O(n)$, or a total of $O(n + w)$ bits for the entire array, the overhead for
breaking it into subarrays, and the binary tree.

A \call{Vertices} query in our structure takes time $O(\log (n/w)) \subset
O(\log n)$ time to find the correct subarray.  Then it must perform a linear
scan in the subarray.  The number of points in the subarray is $O(w)$.  While
on average each point requires $O(1)$ bits, in the worst case, they can each
take $O(w)$ bits, so the subarray may have size $O(w^2)$ bits.  Naively, this
will take $O(w^2)$ time to decompress.  With tables of size $\Theta(n)$ bits,
we can decompress $\Theta(\log n)$ bits at a time, thereby decompressing in
$O(w^2 / \log n)$ time.  If $w \in \Theta(\log n)$, we decompress in $O(w) =
O(\log n)$ time.

A \call{SquareOf} query from an uncompressed array in Morton order required
$O(\log w)$ \call{Vertices} queries.  Now, however, there are two cases.  In
the losslessly compressed structure, the square size is always small: we need
only try $O(\log f_0)$ \call{Vertices} queries before we find the quadtree size.
In the lossily compressed structure, we explicitly store the square size, which
means we can reproduce the quadtree square of a point simply by searching for
the point itself.

\begin{figure}
\begin{tabular}{rl}
$\call{SquareOf}(p)$	&	$O((w^2/\log n) + \log n)$ \\
$\call{Vertices}(s)$	&	$O((w^2/\log n) + \log n)$ \\
$\call{Neighbour}(s, i)$&	$O(1)$ \\
$\call{Child}(s, i)$	&	$O(1)$ \\
$\call{Voronoi}(p)$	&	$O((w^2/\log n) + \log n)$
\end{tabular}
\caption{\label{fig:compress-runtimes} Runtimes for the compressed structure.}
\end{figure}

\begin{theorem}
After preprocessing, we can store a well-spaced set of points using $O(n)$ bits
while supporting the quadtree operations in the times listed in
Figure~\ref{fig:compress-runtimes}.
\end{theorem}

\section{Dynamic Operation}
\label{sec:dynamic}

So far we have not discussed the construction of our structures.  Inserting
a vertex requires a semi-dynamic compression format.  The structure of the
prior section is a balanced tree with fat leaves storing subarrays of $\Theta(w)$
vertices.  To support insertions, we let the size of the subarrays vary
between $w$ and $2w$.  Upon adding a vertex $v$, we find the appropriate
subarray and add $v$ to it in the appropriate spot, encoding it and
re-encoding its successor.  If $v$ becomes the head of the subarray, we
update the dictionary appropriately.  Finally, if the new subarray is too
large, we split it in two equal halves.

\begin{lemma}
We can insert a vertex into our compressed structure in time $O(w + \log n)$.
\end{lemma}

\begin{figure}
\begin{algorithmic}[1]
\item[\call{Read}(file)]
\STATE open the file
\STATE let $n$ be the number of points in the file
\STATE let $P$ to the compressed set of $n$ points all at the origin with height $w$
\STATE let $B$ to be a length $n$ vector of $\log \gamma$-bit integers with initial value 0
\FOR{$i = w-1$ down to $0$}
  \STATE reopen the file
  \STATE $P' \gets \emptyset$
  \FOR{point $j = 1$ to $n$}
    \IF{$B[j]$ is $\gamma$}
      \STATE copy the $j$th point of $P$ to $P'$
    \ELSE
      \STATE Read $p_j$ from the file.
      \STATE $p'_j \gets$ zero out the $i$ least significant bits of $p_j$
      \STATE add $p'_j$ with height $w$ to $P'$
    \ENDIF
  \ENDFOR
  \FOR{point $j = 1$ to $n$}
    \STATE Read the point $p'_j$ and height $h'_j$ from $P'$
    \STATE Check if one of the neighbours of the square $(p'_j, h'_j)$ contains a point
    \STATE If not, increment $B[j]$ to maximum value $\gamma$.
  \ENDFOR
  $P \gets P'$
\ENDFOR
\STATE{\textbf{return} $P$}
\end{algorithmic}
\caption{\label{fig:read-algo}  Initializing the compressed point set in $w$ scans of an input file.}
\end{figure}

This assumes the input was already rounded (or that we are not rounding).
Efficiently rounding is a difficult problem: we need to know the balanced
quadtree cell that contains a point $p$ to know how much information to
discard.  But we can't know that until we have read the entire file into
memory, which takes too much space.  The solution to this problem is to scan
the input file repeatedly.  In the first scan, we read only the first (most
significant) bit of each coordinate.  Having read the file once, if we now have
enough information to compute the balanced quadtree leaf of any point, we no
longer need to read any more bits of that point, but we still need to read more
bits of the other points.  Therefore, we mark the points either as balanced or
not using a bitvector $B$.  In each subsequent scan, we simply copy the points
that we know sufficiently well, while we read an additional bit of the points
we still need to refine.  In the algorithm of Figure~\ref{fig:read-algo}, $B$
is widened to count up to $\gamma$ to allow for less aggressive rounding.
Obviously, in practice, one would want to read multiple bits at a time, to
reduce the number of scans.

\begin{theorem}
\label{thm:dynamic}
We can construct a compressed query structure in $w$ scans of the input file,
in time $O(nw (\log(n) + w^2/\log n))$, never using more than $O(n)$ bits of
memory.  For $w = \Theta(\log n)$ this is $O(n \log^2 n)$ time.
\end{theorem}

\begin{proof}
Each scan takes linear time to read the file, then performs $O(1)$
\call{Vertices} queries per vertex read, each of which has the ugly runtime
presented in the prior section.  Encoding, to add the new point to $P'$, has
equal runtime.
\end{proof}

Given the spatial locality between the queries, it seems likely that one could
improve upon this using finger searching data structures and caching the
decompression.  An alternate improvement would be to amortize the $w^2$-bit
worst-case subarray decompression time by the fact that we only have a total of
$O(n)$ bits.  I conjecture it should be possible to use only $O(nw)$ time over
$w$ scans without violating the space bound.

\section{Compressed Mesh Refinement}
\label{sec:refine}

The previous sections assumed the mesh was given as a well-spaced set of
points.  But where did this data come from, if it doesn't fit in memory?
Here I show how to construct a well-spaced superset (i.e., a quality mesh)
of an ill-spaced input by adding a minimal number of Steiner
points~\cite{bern94provably, ruppert95delaunay}.
I show that the Steiner points are asymptotically free to store, and
that we can compute them using just the queries for our compressed quadtree.

\begin{figure}[t]
\begin{algorithmic}[1]
\item[\call{Refine}($P$, $\rho$)]
\STATE $r \gets 0$
\STATE $M \gets P$
\WHILE{$r$ finite}
  \STATE $r' \gets \infty$
  \FOR{$v \in M$}
    \STATE $s_v \gets \call{SquareOf}(v)$
    \IF{$|s_v| > r$}
      \STATE $r' \gets \min(r', |s_v|)$
    \ELSE
      \STATE compute $2\rho$-clipped Voronoi of $v$
      \WHILE{aspect ratio of $v$ exceeds $\rho$}
        \STATE $r' \gets r$
        \STATE choose $x \in$ clipped Voronoi with $|vx| \geq \rho \NN(v)$
        \STATE round $x$ as for geometry compression
        \STATE insert $x$ into $M$
        \STATE recompute $2\rho$-clipped Voronoi of $v$
      \ENDWHILE
    \ENDIF
  \ENDFOR
  \STATE $r \gets r'$
\ENDWHILE
\end{algorithmic}
\caption{\label{fig:refine} Mesh refinement using a compressed mesh.}
\end{figure}

At heart, the algorithm I propose iterates over each vertex $v$, inspects
its Voronoi cell, then, if it has bad aspect ratio, inserts Steiner points
in appropriate locations.  It terminates when all vertices have good aspect
ratio.  The algorithm chooses Steiner points that are far from any current
vertex: they are far from $v$, but within its Voronoi cell, so they are far
from all other vertices also.  It then rounds the new vertex, which means
that the new vertex has nearest neighbour distance at least $\rho'\NN(v)$,
where $\rho'$ depends on the rounding parameter $\gamma$ and on the dimension
$d$.  This guarantees that after each pass, the smallest nearest neighbour
distance of any bad aspect ratio vertex increases geometrically; it also
guarantees size optimality, using a proof by Ruppert~\cite{ruppert95delaunay}.
Thus we terminate with a minimal-size well-spaced superset in $O(\log_{\rho'}
n)$ rounds.  The caller must set $\rho$ and $\gamma$ as follows to
guarantee termination and size optimality.

\begin{lemma}
Given a point $x$ chosen to be at distance $\NN(x) = \rho \NN(v)$ from its
nearest neighbour $v$, then after rounding $x$ to $x'$ as if $x$ had quadtree
height equal to $v$, then $\NN(x') \geq (\rho - 2^{-\gamma}) \NN(v)$.  Setting
$\rho \geq 1 + 2^{-\gamma}$ guarantees termination.
\end{lemma}

\begin{proof}
The diameter of the quadtree square that contains $v$ is at most $\NN(v)$.
Keeping an additional $\gamma$ bits means the distance $|xx'|$ is at most
$2^{-\gamma} \NN(v)$.  Thus $\NN(x') \geq \rho \NN(v) - |xx'| \geq (\rho -
2^{-\gamma}) \NN(v)$.  The proof of Ruppert~\cite{ruppert95delaunay},
generalized to arbitrary dimensions and specialized to point clouds, states
that we need $\rho - 2^{-\gamma} \geq 1$.
\end{proof}

For efficient operation, I use the \emph{$\beta$-clipped Voronoi} cell
definition of Hudson and \turkoglu{}~\cite{hudson08efficient}.  For a
vertex $v$, the $\beta$-clipped Voronoi cell is the part of the Voronoi
cell that it within distance $\beta \NN(v)$.  We used this result earlier,
but in a well-spaced mesh, the clipped Voronoi cell is precisely the
Voronoi cell.  Theorem 7 of Hudson and \turkoglu{}~\cite{hudson08efficient}
gives conditions under which we can efficiently compute the $\beta$-clipped
Voronoi cell even if not all of the mesh is well-spaced; fundamentally, it
is that all vertices of substantially lesser nearest neighbour distance are
well-shaped.  This motivates the check in the code of Figure~\ref{fig:refine}
to avoid processing vertices whose quadtree square is large.

\begin{theorem}
On a $w$-bit machine, we can compute a minimal-sized well-spaced
superset of an $n$-point input after scanning the input once.  This uses
$O(nw)$ bits of storage.  We can also round the $n$-point input if we
scan the input $w$ times and reduce the storage to $O(m)$
bits, where $m$ is the number of points in the minimal well-spaced superset.
In either case, the procedure takes $O(nw^2 \polylog n)$ time.
\end{theorem}

\begin{proof}
In the first case, the storage consists of the $n$ input points written
in longhand and thus taking $O(nw)$ bits, then the $m$ Steiner points written
in compressed form.  In the worst case, $m \in n \log \Delta$, where $\Delta$
is the spread of the input and on integer input is polynomial in the word size.
Thus the $m$ Steiner points, which take $O(1)$ bits on average, only increase
the size by $O(nw)$.

In the second case, we
use Theorem~\ref{thm:dynamic} to read the input.  The theorem does not
quite apply: geometry compression uses a number of bits linear in the
size of the balanced quadtree, and an
ill-shaped input can define a large quadtree.  But this large quadtree has
size $\Theta(m)$ according to long-standing proofs in the meshing literature
(namely, that the Bern-Eppstein-Gilbert mesh, based on a quadtree, has size
within a constant factor of optimal, as does the Ruppert mesh that we are
computing here).  Therefore, after apply Theorem~\ref{thm:dynamic}, we have
used $O(m)$ bits.  Adding the $m$ rounded Steiner points does not
asymptotically increase the requirement.
\end{proof}

\section{Conclusions}

In the theoretical setting, I see two main directions for extending this work
in the near term.  The first is to consider the case of the weighted Delaunay
triangulation, which is needed for surface reconstruction in higher dimensions
and also for the mesh refinement problem.  I conjecture that all the same
bounds will hold.  There is a caveat: we need to represent the weight.  I
suspect a bounded number of bits suffice for this.  

The second direction reminds the reader that I aimed this work to problems in
the scientific computing setting.  Then the mesh should be holding values (say,
the pressure and velocity).  Can we compress those values?  Typically it is
safe to assume that the values come from a Lipschitz function---that is, that
the derivatives are bounded---so this is a somewhat reasonable hope.

The big remaining question, of course, has to do with practicality: is the
space improvement sufficient in practice to justify the runtime cost?  An early
implementation sees storage costs of 14 bits per vertex (bpv) for the Stanford
bunny, or 18 bpv for the Stanford Lucy dataset, as opposed to 96 bpv when
stored uncompressed in single-precision.  Storing an additional five bits
per vertex bounds the maximum relative error in inter-point distances to
about 10\%, while still offering a factor of three space savings.
In prior succinct data structures work, such a savings was
enough that memory hierarchy effects paid for the runtime cost of
decompressing.  Here we need to pay some additional asymptotic runtime cost as
well; to answer this question we will need to see a full implementation of
an end-to-end system that reads in a file, compresses it to memory, then meshes
it or reconstructs a surface from it.  It may well be that compression of the
style explained here would be well-suited to the distributed setting, where
network transfer costs frequently dominate.

\bibliographystyle{alpha}
\bibliography{main}

\end{document}